\documentclass[a4paper,12pt]{article}

\usepackage{amsthm}{\normalsize }
\usepackage{amsmath}
\usepackage{amssymb}
\usepackage{latexsym}
\usepackage{booktabs}
\usepackage{multirow}
\usepackage{bm}
\usepackage{array}

\usepackage{float}
\usepackage{diagbox} 
\usepackage{threeparttable}
\usepackage[textwidth=18cm,textheight=20cm]{geometry}
\usepackage{natbib}
\newtheorem{theorem}{Theorem}[section]
\newtheorem{proposition}[theorem]{Proposition}
\newtheorem{lemma}[theorem]{Lemma}
\newtheorem{corollary}[theorem]{Corollary}
\newtheorem{remark}[theorem]{Remark}

\newtheorem{example}[theorem]{Example}

\usepackage{graphicx} 

\title{Portfolio analysis with mean-CVaR and mean-CVaR-skewness criteria based on mean--variance mixture models}

\author{\normalsize{$\text{Nuerxiati Abudurexiti}^a, \text{Kai He}^a, \text{Dongdong Hu}^a,\text{Svetlozar T. Rachev}^b, $}\\ \normalsize{$\text{Hasanjan Sayit}^a,  \text{Ruoyu Sun}^a$}\\ \footnotesize{$^a \text{Xi'an Jiaotong Liverpool University, Suzhou, China}$}\\ \footnotesize{$^b \text{Department of Mathematics and Statistics, Texas Tech University, Lubbock, TX, USA}$}}
\date{September 20, 2021}

\begin{document}

\maketitle


\begin{abstract}
The paper \cite{Zhao_Shangmei_And_Lu_Qing_And_Han_Liyan_And_Liu_Yong_And_Hu_Fei_2015} shows that  mean-CVaR-skewness portfolio optimization problems based on asymetric Laplace (AL) distributions can be transformed into  quadratic optimization problems for which closed form solutions can be found. In this note, we show that such a result also holds for mean-risk-skewness portfolio optimization problems when the underlying distribution belongs to a larger class of normal mean--variance mixture (NMVM) models than the class of AL distributions.
We then study the value at risk (VaR) and conditional value at risk (CVaR) risk measures of portfolios of returns with NMVM distributions.
They have closed form expressions for portfolios of normal and more generally  elliptically distributed returns, as discussed in \cite{Rockafellar_R_Tyrrell_And_Uryasev_Stanislav_2000} and in \cite{Landsman_And_Valdez_2003}. When the returns have general NMVM distributions,  these risk measures do not give closed form expressions. In this note, we give approximate closed form expressions for the VaR and CVaR of portfolios of returns with NMVM distributions.
Numerical tests show that our closed form formulas give accurate values for VaR and CVaR and shorten the computational time for portfolio optimization problems associated with VaR and CVaR considerably.

\end{abstract}

\vspace{0.1in}

\section{Introduction}



Numerous empirical studies of asset returns suggest that their distributions deviate from the normal distribution, see \cite{Cont_Rama_And_Tankov_Peter_2004} and  \cite{Schoutens_Wim_2003}. In fact, it has been demonstrated in many papers that
asset returns have fat tails and skewness, see \cite{Rachev_And_Stoyanov_And_Biglova_And_Fabozzi_2005} and  Campbell,
\cite{Lo_And_MacKinlay_1997} for example. The leptekortic features of empirical asset return data suggest that there are more realistic models than the normal distribution. It has been empirically demonstrated in numerous papers that the multivariate Generalized Hyperbolic (mGH) distributions and their sub-classes can describe multivariate financial data returns very well , see \cite{McNeil_Alexander_J_And_Frey_Rudigger_And_Embrechts_Paul_2015} and the references therein. The class of mGH distributions is a subclass of general normal mean--variance mixture (NMVM) models.

A return vector $X=(X_1, X_2, \cdots, X_n)^T$ (here and from now on $T$ denotes transpose) of $n$  assets follows an NMVM distribution if
\begin{equation}\label{GH}
X \overset{d}{=} \mu+\gamma Z+\sqrt{Z}AN_n,
\end{equation}
where $N_n$ is an $n-$dimensional standard normal random variable $N(0, I)$ (here $I$ is the $n$-dimensional identity matrix),  $Z\geq 0$ is a non-negative scalar-valued random variable which is independent of $N_n$, $A\in \mathbb{R}^{n\times n}$ is a matrix, $\mu\in R^n$ and $\gamma\in R^n$ are constant vectors which describe the location and skewness of $X$ respectively.

In general, $Z$ in (\ref{GH}) can be any non-negative valued random variable. But if $Z$ follows a Generalized Inverse Gaussian (GIG) distribution, the distribution of $X$ is called an mGH distribution, see \cite{McNeil_Alexander_J_And_Frey_Rudigger_And_Embrechts_Paul_2015} and \cite{Prause_Karsten_And_Others_1999} for the delails of the GH distributions. When $Z$ is an exponential random variable  with parameter $\lambda=1$, $X$ follows an assymetric Laplace (AL) distribution, see \cite{Mittnik_Stefan_And_Rachev_Svetlozar_T_1993}, \cite{Kozubowski_Tomasz_J_And_Rachev_Svetlozar_T_1994}, and \cite{Kozubowski_Tomasz_J_And_Podgorski_Krzysztof_2001} for the definition and financial applications of AL distributions.

A GIG distribution $W$ has three parameters $\lambda, \chi,$ and $\psi$ and its density is given by
\begin{equation}\label{GIG}
f_{GIG}(w; \lambda, \chi, \psi)=\left \{ \begin{array}{cc}
\frac{\chi^{-\lambda}(\chi \psi)^{\frac{\lambda}{2}}}{2K_{\lambda}(\sqrt{\chi \psi})}w^{\lambda-1}
e^{-\frac{\chi w^{-1}+\psi w}{2}  }, & w>0, \\
0,     & w\le 0,\\
\end{array}
\right.
\end{equation}
where $K_{\lambda}(x)=\frac{1}{2}\int_0^{\infty}y^{\lambda -1}e^{(-\frac{x(y+y^{-1})}{2})}dy$ is the modified Bessel function of the third kind with index $\lambda$ for $x>0$. The parameters in (\ref{GIG}) satisfy $\chi>0$ and $\psi\geq 0$ if $\lambda<0;$ $\chi>0$ and $\psi>0$ if $\lambda=0;$ and $\chi\geq 0$ and $\psi>0$ if $\lambda>0$. The expected value of $W$ is given by
\begin{equation}
E(W)=\frac{\sqrt{\chi/\psi}K_{\lambda+1}(\chi \psi)}{K_{\lambda}(\sqrt{\chi \psi})}.
\end{equation}
With $Z\sim GIG$ in (\ref{GH}), the density function of $X$ has the following form
\begin{equation}
f_X(x)=\frac{(\sqrt{\psi/\chi})^{\lambda}(\psi+\gamma^T\Sigma^{-1}\gamma)^{\frac{n}{2}-\lambda}}{(2\pi)^{\frac{n}{2}}|\Sigma|^{\frac{1}{2}} K_{\lambda}(\sqrt{\chi \psi})}\times \frac{K_{\lambda-\frac{n}{2}}(\sqrt{(\psi+Q(x))(\psi+\gamma^T\Sigma^{-1}\gamma)}e^{(x-\mu)^T\Sigma^{-1}\gamma}}{(\sqrt{(\chi+Q(x))(\psi+\gamma^T\Sigma^{-1}\gamma)})^{\frac{n}{2}-\lambda}},
\end{equation}
where $Q(x)=(x-\mu)^T\Sigma^{-1}(x-\mu)$ denotes the Mahalanobis distance.

The GH distribution contains several special cases that are used frequently in financial modelling. For example, the case $\lambda=\frac{n+1}{2}$ corresponds to a multivariate hyperbolic distribution, see \cite{Eberlein_Ernst_And_Keller_Ulrich_1995} and \cite{Bingham_NICHOLAS_H_And_Kiesel_Rudigger_2001} for applications of this case in  financial modelling. When $\lambda=-\frac{1}{2}$, the distribution of $X$ is called a Normal Inverse Gaussian (NIG) distribution and \cite{Barndorff-Nielsen_Ole_E_1997} proposes NIG as a good model for finance. If $\chi=0$ and $\lambda>0$, the distribution of $X$ is a Variance Gamma (VG) distribution, see \cite{Madan_Dilip_B_And_Seneta_Eugene_1990} for the details of this case. If $\psi=0$ and $\lambda<0$, the distribution of $X$ is called the generalized hyperbolic Student $t$ distribution and \cite{Aas_Kjersti_And_Haff_Ingrid_Hobaek_2006} shows that this distribution matches the empirical data very well.


The mean-risk portfolio optimization problems based on NMVM models were discussed in the recent paper \cite{Shi_And_Kim_2021}. They provided a method that transforms a high dimensional problem into a two-dimensional problem in their proposition 3.1. This enabled them to solve portfolio optimization problems associated with NMVM distributed returns more efficiently. When the returns have Gaussian distributions or more generally  elliptical distributions, the VaR and CVaR risk measures have  closed form expressions for portfolios of such returns, as shown in \cite{Owen_And_Rabinovitch_1983} and \cite{Landsman_And_Valdez_2003}. For returns with more general NMVM distributions, these risk measures do not give closed form formulas. One of the goals  of this paper is to provide approximate closed form formulas for risk measures for portfolios of NMVM distributed returns.

\begin{remark} We remark here that the NMVM models, which are more general than the mGH models, allow for models obtained by the multiple subordination  techniques applied in \cite{ShirvaniM_2021} and \cite{ShirvaniM_2021} for example.  The distributional class of  mGHs has a specific structure of the subordinators and thus fixes the shape of the distributional  tail behavior. NMVM includes the class of  all distributions representing the unit increment of multivariate subordinated L\'evy processes, see \cite{Shirvani_2021} for example. The number of subordinations changes the behavior of the tails of the distribution and can be used as a model parameter. In particular, the class of NMVM models includes multiple subordinated mGH. For example, if the first subordination determines the heavy tailedness and the skewness of the stock returns, the second subordination can be used to track transaction volume, and the third the ESG score of the company issuing the stock.  As shown in \cite{ShirvaniM_2021} and \cite{Shirvani_2021},
one subordination is not necessarily  sufficient in financial modeling. These facts necessitate the need for the discussion of the general class of NMVM models.
\end{remark}

Measuring the risk of a financial position is a complex process. It is standard to assess the riskiness of a financial position by means of convex risk measures, see \cite{Frittelli_Marco_And_Gianin_Emanuela_Rosazza_2002}, \cite{Follmer_Hans_And_Schied_Alexander_2002}, \cite{Artzner_PhilippeAndDelbaen_FreddyAndEber_Jean-MarcAndHeath_David_1999}. A convex risk measure is a convex function  $\rho: L^P\rightarrow (-\infty, +\infty]$ which is
\begin{enumerate}
\item [(a)] Monotone: $\rho(X)\le \rho(Y)$ if $X\geq Y$ almost surely.
\item [(b)] Cash-invariant:  $\rho(X+m)=\rho(X)-m$    for all $m\in R$.
\end{enumerate}
Here $L^p$ denotes the class of random variables with finite $p$ moments (here, $p\in [1, +\infty)$). A convex risk measure $\rho$ satisfies the convexity property $\rho(\lambda X+(1-\lambda)Y)\le \lambda \rho(X)+(1-\lambda)\rho(Y), \lambda \in [0, 1]$. A convex risk measure $\rho$ is called coherent if it satisfies the positive homogeneity property $\rho(\lambda X)=\lambda \rho(X)$  for any $\lambda>0$. The quantity $\rho(X)$ can be viewed as the minimal capital that has to be added to the financial position $X$ to make it acceptable. A risk measure $\rho$ with the above two properties (a) and (b) is called a monetary risk measure, see \cite{Follmer_Hans_And_Schied_Alexander_2002}. Monetary risk measures are also defined through an associated acceptable set.  A subset $\mathcal{A}$ of $\mathcal{X}$ is called acceptable if  $X\in \mathcal{A}$ and $Y\geq X$ implies $Y\in \mathcal{A}$. A risk measure
$\rho(X)$ associated with an acceptable set $\mathcal{A}$ is defined to be the minimum capital that has to be added to $X$ to make it acceptable, i.e.,
\begin{equation}\label{5}
\rho(X):=\inf\{a\in R| X+a\in \mathcal{A}\}.
\end{equation}
A risk measure can also be defined by (\ref{5}) for a given acceptable set $\mathcal{A}$. For any risk measure $\rho$, its associated acceptance set is given by $\mathcal{A}_{\rho}=\{X\in \mathcal{X}| \rho(X)\le 0\}$. The risk measure $\rho$ can be recovered from $\mathcal{A}_{\rho}$ as in (\ref{5}). The concept of coherent risk measures  was first introduced in the seminal paper \cite{Artzner_PhilippeAndDelbaen_FreddyAndEber_Jean-MarcAndHeath_David_1999}, also see \cite{Malevergne_Yannick_And_Sornette_Didier_2006}.  Convex risk measures were introduced and studied in \cite{Frittelli_Marco_And_Gianin_Emanuela_Rosazza_2002}, \cite{Heath_David_2000}, and \cite{Follmer_Hans_And_Schied_Alexander_2002}.


A monetary risk measure $\rho$ is law invariant if $\rho(X)=\rho(Y)$ whenever $X$ and $Y$ are equal in distribution. The value at risk (VaR) is a law invariant  monetary risk measure. For any $\alpha\in (0, 1)$, the value at risk at significance level $\alpha$ is denoted by $VaR_{\alpha}$ and its acceptance set is given by
\[
\mathcal{A}_{var_{\alpha}}=\{X\in \mathcal{X}| P(X<0)\le \alpha\}.
\]
The VaR is given by
\begin{equation}
\begin{split}
VaR_{\alpha}(X)&=\inf \{m\in R| X+m\in \mathcal{A}_{var_{\alpha}}\}\\
&=-\inf \{c\in R| P(X\le c)> \alpha\}=-q_X^+(X),\\
\end{split}
\end{equation}
where $q_X^+(\alpha)$ is the upper $\alpha$-quantile of $X$. As pointed out in \cite{Artzner_PhilippeAndDelbaen_FreddyAndEber_Jean-MarcAndHeath_David_1999}, the VaR lacks the subadditivity property. In particular, VaR is a non-convex function and portfolio optimization problems with it lead to multiple local extremes. Therefore portfolio optimization with VaR is computationally expensive.

Recently, the risk measure CVaR has become popular both in academia and in finance. The CVaR is a coherent risk measure (see, e.g., \cite{Acerbi_Carlo_And_Tasche_Dirk_2002}) and therefore has some favorable  properties that VaR lacks. The conditional value at risk (average value at risk, tail value at risk, or expected shortfall) is defined for any $\alpha \in (0, 1)$ as follows
\[
CVaR_{\alpha}(X)=:-\frac{1}{\alpha}\int_0^{\alpha}Var_{\beta}(X)d\beta=-\frac{1}{\alpha}E[X1_{\{X\le -VaR_{\alpha}(X)\}}].
\]
Since $VaR_{\alpha}$ is decreasing in $\alpha$, we clearly have $CVaR_{\alpha}(X)\geq VaR_{\alpha}(X)$. In fact, CVaR is the smallest coherent risk measure which is law invariant and dominates the VaR.

It is well known that a convex measure $\rho$ is law invariant if and only if
\begin{equation}\label{cv1}
\rho(X)=\sup_{Q\in \mathcal{M}_1((0, 1])}\big [\int_{(0, 1]} CVaR_{\alpha}(X)Q(d\alpha)-\beta^{min}(Q)\big ],
\end{equation}
where
\begin{equation}\label{cv2}
\beta^{min}(Q)=\sup_{X\in \mathcal{A}_{\rho}}\int_{[0, 1]}CVaR_{\alpha}(X)Q(d\alpha),
\end{equation}
and $\mathcal{M}_1((0, 1])$ is the space of probability measures on $(0, 1]$. A coherent risk measure $\rho$ is law invariant if and only if
\begin{equation}\label{cv3}
\rho(X)=\sup_{Q\in \mathcal{M}}\int_{(0, 1]}CVaR_{\alpha}(X)Q(d\alpha)   \end{equation}
for a subset $\mathcal{M}$ of $\mathcal{M}_1((0, 1])$. For the details of these results see \cite{Kusuoka_Shigeo_2001}, \cite{Dana_Rose-Anne_2005}, and \cite{Frittelli_Marco_And_Gianin_Emanuela_Rosazza_2005}.

\begin{remark} We remark here that the relations (\ref{cv1}), (\ref{cv2}), and (\ref{cv3}) above show that any law invariant convex risk measure can be expressed in terms of CVaR. In Section 3 of this paper, we will discuss approximate closed form formulas for CVaR for portfolios of NMVM distributed returns. One can then use these formulas in place of CVaR in (\ref{cv1}), (\ref{cv2}), and  (\ref{cv3}) above and construct simpler expressions for any law invariant convex risk measure.
\end{remark}

This paper is organized as follows. In Section 2, we study a mean-risk-skewness multi-objective optimization problem and provide closed form solutions for optimal portfolios when the return vectors follow a certain class of NMVM models.
In Section 3, we give approximate closed form expressions for the risks of the portfolios of NMVM returns when the risk measures are coherent ones. In Section 4, we present the results of numerical tests of our results.

\section{Closed form solutions for mean-risk-skewness optimal portfolios}
As mentioned in the Abstract,  \cite{Zhao_Shangmei_And_Lu_Qing_And_Han_Liyan_And_Liu_Yong_And_Hu_Fei_2015} studies  mean-CVaR-skewness optimization problems for AL distributions. The class of AL distributions considered in their paper have the mean--variance mixture form
\begin{equation}\label{xandy}
X=\gamma Z+\tau \sqrt{Z}N
\end{equation}
for two real numbers $\gamma \in R$ and $\tau>0$, where $Z\sim Exp(1)$, and it is independent of the standard normal random variable $N$. The characteristic function of $X$ is given by the formula (1) in their paper with $\mu$ replaced by $\gamma$. In this section and for the rest of the paper, we use ``$\overset{d}{=}$'' to denote the equivalence in distribution of two random variables. ``Skew'' denotes the skewness of a random variable and ``Kurt'' denotes the kurtoses of a random variable.

In the multi-dimensional case, the AL distributions they have considered have the form $X=\gamma Z+\sqrt{Z}AN_n$, where now $\gamma \in R^n$ and $A$ is an $n\times n$ matrix and $N_n$ is a $n-$dimensional standard normal random variable. For any portfolio $\omega=(\omega_1, \cdots, \omega_n)^T$, they calculated $CVaR_{\alpha}(\omega^TX)$
and $SKew(\omega^TX)$ as
\begin{equation}
\begin{split}
CVaR_{\alpha}(\omega^TX)=&(1-\ln \alpha)g(\gamma_p, \tau_p)-g(\gamma_p, \tau_p)\ln (2+\frac{\gamma_p}{g(\gamma_p, \tau_p)}),\\
Skew(\omega^TX)=&\frac{2\gamma_p^3+3\gamma_p\tau_p^2}{(\gamma_p^2+\tau_P^2)^{\frac{3}{2}}},\\
\end{split}
\end{equation}
where $\gamma_p, \tau_p,$ and $g(\mu_p, \tau_p)$ are given as in (9), (11), and in (16) in their paper respectively. Note that here $\gamma_p$ represents $\mu_p$ in their paper.

In their paper, they considered the mean-CVaR-skewness multi-objective problem
\begin{equation}\label{mrhoso}
\left \{
\begin{aligned}
\min_{\omega\in R^n}&\;\;\;CVaR_{\beta}(\omega^TX),\\
 \max_{\omega \in R^n}&\;\;\;Skew(\omega^TX),\\
 s.t.&\;\;\; E(\omega^TX)=r,\; \omega^Te=1.
\end{aligned}
\right.
\end{equation}
for any return $r>0$. This problem can be written under general law invariant risk measure $\rho$ as follows:
\begin{equation}\label{mrhoso1}
\left \{
\begin{aligned}
\min_{\omega\in R^n}&\;\;\;\rho(\omega^TX),\\
 \max_{\omega \in R^n}&\;\;\;Skew(\omega^TX),\\
 s.t.&\;\;\; E(\omega^TX)=r,\; \omega^Te=1.
\end{aligned}
\right.
\end{equation}
\begin{remark} We remark here that the problem (\ref{mrhoso1}) is a multi-objective problem, i.e., optimize both $\rho(\omega^TX)$ and $Skew(\omega^TX)$ simultaneously under the constraints $ E(\omega^TX)=r,\; \omega^Te=1.$ In comparison, the relevant paper \cite{Ararat_2020} discusses a static portfolio optimization problem also and their optimization problem is a single-objective problem.
\end{remark}
\begin{remark} Such mean-risk-skewness portfolio optimization problem as in (\ref{mrhoso1}) have been studied in many papers in the past, see \cite{Hiroshi_Ken_1995}, \cite{Zhao_Shangmei_And_Lu_Qing_And_Han_Liyan_And_Liu_Yong_And_Hu_Fei_2015}, and the references therein, for example. This is because, as stated in the introduction of \cite{Hiroshi_Ken_1995}, many investors prefer a positively skewed distributions to a negative one and
also larger skewed distributions compared to smaller skewed distributions if the mean and the risk are the same.
\end{remark}

In their paper, they observe that CVaR, as was shown in the proof of their Theorem 2, is an increasing function of $\tau_p$ and Skew is a decreasing function of $\tau_p$. They concluded, therefore, that minimizing $\tau_p$ gives the same solution as (\ref{mrhoso}) above and also it translates into the quadratic optimization problem (17) in their paper.

In this paper, we take a different approach than \cite{Zhao_Shangmei_And_Lu_Qing_And_Han_Liyan_And_Liu_Yong_And_Hu_Fei_2015}
for the solution of the problem (\ref{mrhoso}) above and obtain similar results for $X=\gamma Z+\sqrt{Z}AN_n$ when $Z$ is more general than $Exp(1)$ and also when CVaR is replaced by a general law-invariant risk measure as in (\ref{mrhoso1}).

We first fix some notation. We denote by $A_1^c, A_2^c, \cdots, A_n^c$ the column vectors of $A$. We assume that $A$ is an invertable matrix. We  write $\gamma$ as linear combinations of $A_1^c, A_2^c, \cdots, A_n^c$, i.e.,  $\gamma=\sum_{i=1}^n\gamma_iA_i^c$. We denote by $\gamma_0=(\gamma_1, \gamma_2, \cdots, \gamma_n)^T$ the vector of the corresponding coefficients of such linear combination. As usual, we denote by $||x||=\sqrt{x_1^2+x_2^2+\cdots+x_n^2}$ the Euclidean norm of the vector $x$. Also, for each NMVM distribution $X$ as above, we define
\begin{equation}\label{y}
Y=\gamma_0Z+\sqrt{Z}N_n,
\end{equation}
and we call $Y$ the NMVM distribution associated with $X$.

For any portfolio $\omega$ we define a transformation $\mathcal{T}: R^n\rightarrow R^n$ by $\mathcal{T}\omega=x$, where $x=(x_1, x_2, \cdots, x_n)^T$ is given by
\begin{equation}\label{T}
x_1=\omega^T A_1^c,\; x_2=\omega^T A_2^c,\; \cdots,\; x_n=\omega^T A_n^c.  \end{equation}
In matrix form, this is written as $x^T=\omega^TA$. From now on,  for any domain $D$ of portfolios $\omega$, we denote by $\mathcal{R}_D$ the image of $D$ under $\mathcal{T}$.

We start our discussions with the following simple lemma. In the following lemma,  $N$ denotes a standard normal random variable that independent of the mixing distribution $Z$.
\begin{lemma}\label{lemk}
We have $\omega^TX=x^TY$ and $x^TY\overset{d}{=}x^T\gamma_0Z+\sqrt{x^Tx}\sqrt{Z}N$. Therefore we have $\rho(w^TX)=\rho(x^TY)=\sqrt{x^Tx}\rho(cos[\theta(\gamma_0, x)]Z+\sqrt{Z}N)$ for any law invariant coherent risk measure $\rho$ whenever $\omega\in R^n$ and $x\in R^n$ are related by (\ref{T}). Here $cos[\theta(\gamma_0, x)]$ is the cosine of the angle between $\gamma_0$ and $x$.
\end{lemma}
\begin{proof} Note that $\omega^TX=\omega^T\gamma Z+\sqrt{Z}(\omega^TA)N_n$. We have $\omega^T\gamma=x^T\gamma_0$ and $x=\omega^TA$. Therefore $\omega^TX=x^TY$. Since $x^TN_n\overset{d}{=}\sqrt{x^Tx}N$ and $\rho$ is a law invariant risk measure, we have $\rho(w^TX)=\rho(x^TY)=\rho(x^T\gamma_0Z+\sqrt{x^Tx}\sqrt{Z}N)$. Since $x^T\gamma_0=\sqrt{x^Tx}||\gamma_0||cos[\theta(\gamma_0, x)]$ and $\rho$ is positive homogeneous, we have  $\rho(x^TY)=\sqrt{x^Tx}\rho(cos[\theta(\gamma_0, x)]Z+\sqrt{Z}N)$.
\end{proof}

The above lemma shows that optimization problems under mean-risk-skewness  criteria can be studied in the $x$-coordinate system rather than the original $\omega$-coordinate system. This has some advantage as will be seen shortly.

Below, we write the problem (\ref{mrhoso1}) in the $x$-coordinate system as follows

\begin{equation}\label{mrhos}
\left \{
\begin{aligned}
\min_{x \in R^n}&\;\;\;\rho(x^TY),\\
 \max_{x \in R^n}&\;\;\;Skew(x^TY),\\
 s.t.&\;\;\; x^Tm=r,\; x^Te_A=1,
\end{aligned}
\right.
\end{equation}
where $e_A=A^{-1}e$ and $m=\gamma_0EZ$.

Next, we will show that the problem (\ref{mrhos}) can be reduced to a quadratic optimization problem when $X$ satisfies certain conditions which are satisfied by AL distributions. We first need to calculate $\mbox{SKew}(x^TY)$.

\begin{proposition}\label{skew} For any $x=\omega^TA$,  we have
\begin{equation}\label{1a5}
\begin{split}
\mbox{StD}(x^TY)=&||x||\sqrt{||\gamma_0||^2\phi^2Var(Z)+EZ},\\
\mbox{SKew}(x^TY)=&\frac{||\gamma_0||^3\phi^3m_3(Z)+3||\gamma_0||\phi Var(Z)} {(||\gamma_0||^2\phi^2Var(Z)+EZ)^{\frac{3}{2}}}\\
\mbox{Kurt}(x^TY)=&\frac{||\gamma_0||^4\phi^4m_4(Z)+6||\gamma_0||^2\phi^2\big [EZ^3-2EZ^2EZ+(EZ)^3\big ]+3EZ^2}{\big [||\gamma_0||^2\phi^2Var(Z)+EZ\big ]^2},\\
\end{split}
\end{equation}
where $m_3(Z)=E(E-EZ)^3$, $m_4(Z)=E(Z-EZ)^4$, and $\phi=cos(x, \gamma_0)$.
\end{proposition}

\begin{proof} Note that $x^TY\overset{d}{=}x^Y\gamma_0+||x||\sqrt{Z}N=||x||(||\gamma_0||\phi Z+\sqrt{Z}N)$. Therefore, on putting $Y_0=||\gamma_0||\phi Z+\sqrt{Z}N$, we have $\mbox{StD}(x^TY)=||x||\mbox{StD}(Y_0),\; \mbox{SKew}(x^TY)=\mbox{SKew}(Y_0),$ and $\mbox{Kurt}(x^TY)=\mbox{Kurt}(Y_0)$. Since $Z$ is independent of $N$, a straightforward calculation gives \[
\mbox{StD}(Y_0)=\sqrt{||\gamma_0||^2\phi^2Var(Z)+EZ}.
\]
We have $\mbox{SKew}(Y_0)=E(Y_0-EY_0)^3/ [\mbox{StD}(Y_0)]^{3}$ and again a straightforward calculation shows that $E(Y_0-EY_0)^3$ equals the numerator of $\mbox{SKew}(x^TY)$ in (\ref{1a5}). The kurtosis equals $\mbox{Kurt}(Y_0)=E(Y_0-EY_0)^4/[\mbox{StD}(Y_0)]^4$. It can be easily checked that $E(Y_0-EY_0)^4$ equals the numerator of $\mbox{Kurt}(x^TY)$ in (\ref{1a5}).
\end{proof}

\begin{remark}\label{2a3} From Proposition \ref{skew}, we observe that StD depends on $x$ through both $cos(x, \gamma_0)$ and $||x||$. However, both SKew and Kurt depend on $x$ only through the angle $cos(x, \gamma_0)$. 
\end{remark}

We need to calculate the derivative of $\mbox{SKew}(x^TY)$ with respect to $\phi $. A straightforward calculation gives us
\begin{equation}\label{skewd}
\frac{\partial \mbox{SKew}(x^TY)}{\partial \phi}=\frac{3||\gamma_{0}||^{3}\Big (m_{3}(Z)EZ  -2Var^2(Z)\Big )\phi^{2}+3||\gamma_{0}||Var(Z)EZ}{\left(||\gamma_{0}||^{2}\phi^{2}Var(Z)+EZ\right)^{\frac{5}{2}}}.
\end{equation}

From (\ref{skewd}) we see that if the mixing distribution $Z$ satisfies
\begin{equation}\label{keycc}
m_3(Z)EZ\geq 2Var^2(Z),
\end{equation}
then $\frac{\partial}{\partial\phi}SKew(x^{T}Y)\geq 0$.

\begin{remark} Clearly, if the mixing distribution $Z$ is such that $\frac{\partial}{\partial\phi}SKew(x^{T}Y)\geq 0$ for all $\phi \in [-1, 1]$, a condition which is weaker than (\ref{skewd}),  then SKew is an increasing function of $\phi$. Here we singled out the condition (\ref{keycc}) for its simplicity as a sufficient condition for $\frac{\partial}{\partial\phi}SKew(x^{T}Y)\geq 0$.
\end{remark}

Now we state the main result of this paper.
\begin{theorem}\label{2.5} Let $\rho$ be any law invariant coherent risk measure. Assume $Z$ in (\ref{xandy}) is such that $\frac{\partial}{\partial\phi}SKew(x^{T}Y)\geq 0$ for all $\phi \in [-1, 1]$. Then the solution of the problem (\ref{mrhos})
is equal to the solution of the following quadratic optimization problem:
\begin{equation}\label{quadrat}
\left \{
\begin{aligned}
\min_{x\in R^n}&\;\;\;x^Tx,\\
 s.t.&\;\;\; x^Tm=r,\; x^Te_A=1,\\
\end{aligned}
\right.
\end{equation}
and is given explicitly  by $x=sm+te_A$, where
\begin{equation}\label{mmmm}
s=2\begin{vmatrix}r&e_A^Tm\\ 1&e^T_Ae_A \end{vmatrix}/\begin{vmatrix}m^Tm&e_A^Tm\\m^Te_A&e^T_Ae_A \end{vmatrix}\;\;\;\mbox{and}\;\;\;  t=2\begin{vmatrix}m^Tm&r\\ m^Te_A&1 \end{vmatrix}/\begin{vmatrix}m^Tm&e_A^Tm\\m^Te_A&e^T_Ae_A \end{vmatrix}.
\end{equation}
\end{theorem}
\begin{proof} From Lemma (\ref{lemk}) above we have $\rho(x^TY)=\rho(\frac{x^T\gamma_0}{||\gamma_0||}Z +||x||\sqrt{Z}N)$. Since $m=\gamma_0EZ$ from the constraint (\ref{quadrat}) we have $x^T\gamma_0=r/EZ$. We conclude that
$\rho(x^TY)=\rho(\frac{r}{||\gamma_0||EZ}Z +||x||\sqrt{Z}N)$
is a function of $||x||$ under the constraint (\ref{quadrat}). From part (3) of Theorem 3.1 of \cite{Shi_And_Kim_2021} we conclude that $\rho$ is a non-decreasing function of $||x||$. On the other hand, the condition on Skew ensures that SKew is an increasing function of $\phi=\frac{x^T\gamma_0}{||x||||\gamma_0||}$. With the constraint $x^T\gamma_0=r/EZ$, $\phi$ becomes a decreasing function of $||x||$. Therefore decreasing $||x||$ minimizes $\rho$ and maximizes Skew simultaneously. To show (\ref{mmm}), we form the Lagrangian $\mathcal{L}=x^Tx+s(r-x^Tm)+t(1-x^Te_A)$. The first order condition gives
\begin{equation}\label{el}
\frac{d\mathcal{L}}{dx}=2x-sm-te_A=0.
\end{equation}
From this we get $x^T=\frac{1}{2}sm^T+\frac{1}{2}te_A^T$ and since $x^Tm=r$ and $x^Te_A=1$, we obtain two equations $sm^Tm+te_A^Tm=2r$ and $sm^Te_A+te^T_Ae_A=2$. The solution of these two equations give (\ref{mmmm}).
\end{proof}

For the remainder of this section, we give examples of the mixing distributions $Z$ that satisfy condition (\ref{keycc}).

\begin{example} We consider the case of gamma distributions $G(\lambda, \frac{\gamma^2}{2})$ with the density function
\[
f(x)=(\frac{\gamma^2}{2})^{\lambda}\frac{x^{\lambda-1}}{\Gamma(\lambda)}e^{-\frac{\gamma^2}{2}x}1_{\{x\geq 0\}}.
\]
In this case we have $EZ=\lambda \frac{2}{\gamma^2}, Var(Z)=\lambda (\frac{}2{\gamma^2})^2,$ and $m_3(Z)=2\lambda(\frac{2}{\gamma^2})^3$. We obtain
\[
-2Var^2(Z)+m_{3}(Z)EZ=-2\lambda^{2}\left(\frac{1}{\gamma^{2}}\right)^{4}+2\lambda\left(\frac{1}{\gamma^{2}}\right)^{3}\lambda\left(\frac{1}{\gamma^{2}}\right)=0.
\]
Thus, in this case $Z$ satisfies the condition (\ref{keycc}). Note that $G(\lambda, 1)$ is an $Exp(1)$ random variable and therefore $Exp(1)$ also satisfies the condition (\ref{keycc}). Thus for asymmetric Laplace distributions the problem (\ref{mrhos}) can be reduced to a quadratic optimization problem (\ref{quadrat}). This also shows that our Theorem  \ref{2.5} above extends Theorem 2 in \cite{Zhao_Shangmei_And_Lu_Qing_And_Han_Liyan_And_Liu_Yong_And_Hu_Fei_2015} to the case of more general mean--variance mixture models.
\end{example}
\begin{example} Now consider inverse Gaussian random variables $Z\sim IG(\delta,\gamma)$ with probability density function
\begin{equation*}
\begin{split}
 f_{iG}(z)=\frac{\delta}{\sqrt{2\pi}}e^{\delta\gamma}z^{-\frac{3}{2}}e^{-\frac{1}{2}(\delta^{2}z^{-1}+\gamma^{2}z)}1_{\{z>0\}}.
\end{split}
\end{equation*}
The moments of these random variables are calculated in \cite{Hammerstein_EAv_2010}:
\begin{equation*}
\begin{split}
 E[Z^{r}]=\frac{K_{-\frac{1}{2}+r}(\delta\gamma)}{K_{-\frac{1}{2}}(\delta\gamma)}\left(\frac{\delta}{\gamma}\right)^{r},\quad\text{for $\delta,\gamma>0$},
\end{split}
\end{equation*}
where
\begin{equation*}
K_{-\frac{1}{2}}(x)=K_{\frac{1}{2}}(x)=\sqrt{\frac{\pi}{2x}}e^{x},\;\;
 K_{\lambda+1}(x)=\frac{2\lambda}{x}K_{\lambda}(x)+K_{\lambda-1}(x).
\end{equation*}
From these, we have
\begin{equation*}
\begin{split}
 E[Z]&=\frac{K_{\frac{1}{2}}(\delta\gamma)}{K_{-\frac{1}{2}}(\delta\gamma)}\left(\frac{\delta}{\gamma}\right)=\frac{\delta}{\gamma},\\
 E[Z^{2}]&=\frac{K_{\frac{3}{2}}(\delta\gamma)}{K_{-\frac{1}{2}}(\delta\gamma)}\left(\frac{\delta}{\gamma}\right)^{2}
 =\frac{\frac{1}{\delta\gamma}K_{\frac{1}{2}}(\delta\gamma)+K_{-\frac{1}{2}}(\delta\gamma)}{K_{-\frac{1}{2}}(\delta\gamma)}\left(\frac{\delta}{\gamma}\right)^{2}=\frac{\delta}{\gamma^{3}}+\left(\frac{\delta}{\gamma}\right)^{2},\\
 E[Z^{3}]&=\frac{K_{\frac{5}{2}}(\delta\gamma)}{K_{-\frac{1}{2}}(\delta\gamma)}\left(\frac{\delta}{\gamma}\right)^{3}
 =\frac{\frac{3}{\delta\gamma}K_{\frac{3}{2}}(\delta\gamma)+K_{\frac{1}{2}}(\delta\gamma)}{K_{-\frac{1}{2}}(\delta\gamma)}\left(\frac{\delta}{\gamma}\right)^{3}
 =\left(\frac{3}{\delta\gamma}\frac{K_{\frac{3}{2}}(\delta\gamma)}{K_{-\frac{1}{2}}(\delta\gamma)}+1\right)\left(\frac{\delta}{\gamma}\right)^{3}\\
 &=\left(\frac{3}{\delta\gamma}\left(\frac{1}{\delta\gamma}+1\right)+1\right)\left(\frac{\delta}{\gamma}\right)^{3}=\frac{3\delta}{\gamma^{5}}+\frac{3\delta^{2}}{\gamma^{4}}+\left(\frac{\delta}{\gamma}\right)^{3},
\end{split}
\end{equation*}
and
\begin{equation*}
\begin{split}
 Var(Z)=E[Z^{2}]-(E[Z])^{2}=\frac{\delta}{\gamma^{3}},
\end{split}
\end{equation*}
\begin{equation*}
\begin{aligned}
 m_{3}(Z)=E[Z^{3}]-3E[Z^{2}]E[Z]+2(E[Z])^{3}=\frac{3\delta}{\gamma^{5}}.
\end{aligned}
\end{equation*}
Therefore
\begin{equation*}
\begin{split}
 -2Var^2(Z)+m_{3}(Z)EZ=-2\frac{\delta^{2}}{\gamma^{6}}+3\frac{\delta^{2}}{\gamma^{6}}=\frac{\delta^{2}}{\gamma^{6}}>0.
\end{split}
\end{equation*}
This shows that inverse Gamma random variables also satisfy (\ref{keycc}).
\end{example}
\begin{remark} We should mention that in the above examples we only gave two classes of random variables that satisfy (\ref{keycc}). However, the class of random variables that satisfy (\ref{keycc}) is not restricted to these two types of random variables. One can check that the class of generalized inverse Gaussian random variables GIG also satisfies (\ref{keycc}) for certain parameter values.
\end{remark}

\section{Portfolio optimization under mean-risk criteria}
\subsection{Recent results}
The mean--variance portfolio theory was first introduced by \cite{Markowitz_Harry_1959} and since then it has been  very popular for practitioners and scholars. However, the  Markowitz portfolio theory neglects downside risk. It has been shown in many papers that downside risk can affect returns significantly, see \cite{Bollerslev_And_Todorov_2011} and  \cite{Bali_And_Cakici_2004} for example. Because of this, many alternative risk measures have been the focus of academic research, see \cite{Chekhlov_And_Uryasev_And_Zabarankin_2005}, \cite{Konno_And_Shirakawa_And_Yamazaki_1993}. Among downside risk measures, the VaR and CVaR have been extensively studied, see \cite{Rockafellar_R_Tyrrell_And_Uryasev_Stanislav_2000} and \cite{Kolm_And_Tutuncu_And_Fabozzi_2014}.

The risk measure CVaR was first studied in the context of portfolio optimization problems by \cite{Rockafellar_R_Tyrrell_And_Uryasev_Stanislav_2000,Rockafellar_R_Tyrrell_And_Uryasev_Stanislav_2002}. They  showed  that  a  mean-CVaR  optimization  problem  can be  transformed  into  a  linear  programming problem that improves the efficiency of solving portfolio optimization problems associated with CVaR significantly. However, closed form approximate formulas are still more convenient and efficient and therefore we will discuss such formulas for them below.

For elliptical distributions, the  VaR and CVaR have closed form expressions, as studied in \cite{Landsman_And_Valdez_2003}. For the general class of NMVM models, these risk measures do not give closed form expressions. In this section, we will present approximate closed form formulas for them. Below, we start by discussing the properties of portfolios of NMVM distributed returns.



Let $X$ denote the return of $n$ assets and assume that $X$ has the NMVM distribution
\begin{equation}\label{2ndx}
X=\mu+\gamma Z+\sqrt{Z}AN_n,
\end{equation}
where $\mu \in R^n$ is a constant vector and the other symbols are as in the multi-dimensional case in (\ref{xandy}) above. For any portfolio $\omega = (\omega_1,\omega_2,...,\omega_n)^T$, we have
\begin{eqnarray}\label{omex}
\omega^T X\overset{d}{=}\omega^T \mu+\omega^T\gamma Z+\sqrt{Z}\sqrt{\omega^T\Sigma\omega}N,
\end{eqnarray}
where $N$ is a standard normal random variable and $\Sigma=AA^T$. We will denote the density and distribution functions of the scalar valued random variables $\omega^TX$ by $f_{\omega}$ and $F_{\omega}$ respectively for each portfolio
$\omega$. Then, according to formula $(7)$ of \cite{Hellmich_And_Kassberger_2011}, the conditional value at risk of $-\omega^TX$ is given by
\begin{equation}\label{HK}
CVaR_{\beta}(-\omega^TX)=-\frac{1}{1-\beta}\int_{-\infty}^{F_{\omega}^{-1}(1-\beta)}yf_{\omega}(y)dy,
\end{equation}
where the quantile-function $F_{\omega}^{-1}(\cdot)$ is calculated by a root-finding method and the integral in (\ref{HK}) is calculated by numerical integration. Therefore, optimization problems like
\begin{eqnarray}\label{CVaR}
\min_{\omega\in D }CVaR(-\omega^TX)
\end{eqnarray}
in some domain $D$ of the set of portfolios are time consuming and computationally heavy. Another approach to solving
the problem (\ref{CVaR}) was proposed in \cite{Rockafellar_R_Tyrrell_And_Uryasev_Stanislav_2000,Rockafellar_R_Tyrrell_And_Uryasev_Stanislav_2002}. They introduced an auxiliary function
\begin{equation}
F_{\beta}(\omega, \alpha)=\alpha+\frac{1}{1-\beta}\int_{y\in R^m}[-\omega^Ty-\alpha]^+p(y)dy,
\end{equation}
where $\alpha$ is a real number and $p(y)$ is the $d$-dimensional probability density function of asset returns. They showed that
\begin{eqnarray}
CVaR_{\beta}(-\omega^TX)=\min_{\alpha \in R}F_{\beta}(\omega, \alpha).
\end{eqnarray}
Since $\omega^TX$ follows $f_{\omega}$, the above $F_{\beta}(\omega, \alpha)$ can also be written as
\begin{eqnarray}\label{13}
F_{\beta}(\omega, \alpha)=\alpha+\frac{1}{1-\beta}\int_{x\in R}[-x-\alpha]^+f_{\omega}(x)dx.
\end{eqnarray}
With these, the optimization problem (\ref{CVaR}) becomes
\begin{eqnarray}\label{minCVaR}
\min_{\omega\in D }CVaR(-\omega^TX)=\min_{\omega \in D}\min_{\alpha \in R}F_{\beta}(\omega, \alpha).
\end{eqnarray}
Solving the problem (\ref{minCVaR}) is also time consuming as one needs to calculate (\ref{13}) numerically for each $\omega \in D$, and then for each $\omega \in D$ the right-hand side of (\ref{minCVaR}) needs to be minimized for $\alpha$.

The recent paper by \cite{Shi_And_Kim_2021} studies portfolio optimization problems with NMVM distributed returns. More specifically,  they studied the optimization problem
\begin{equation}\label{15}
\begin{split}
&\min _{\omega \in D}\rho(\omega^TX)\\
s.t.\;\;\; & \omega^Te=1\\
&E[\omega^TX]\geq k
\end{split}
\end{equation}
for any coherent risk measure $\rho$, where $D$ is a subset of the feasible portfolio set,  $k\in R$,  and $e=(1, \cdots, 1)^T$ is an $n$-dimensional column vector in which each component equals one. They showed that the optimal solution $\omega^{\star}$ to (\ref{15}) can be expressed as
\begin{equation}\label{16}
\omega^{\star}=\Sigma^{-1}(\mu, \gamma, e)D^{-1}(\tilde{\mu}^{\star}, \tilde{\gamma}^{\star}, 1)^T,\\
\end{equation}
where
\begin{equation}\label{17}
\begin{split}
(\tilde{\mu}^{\star}, \tilde{\gamma}^{\star})=&\mbox{argmin}_{\tilde{\mu}, \tilde{\gamma}}\rho(\tilde{\mu}+\tilde{\gamma}Z+\sqrt{g(\tilde{\mu}, \tilde{\gamma})Z}N)\\
&\\
\mbox{s.t.}\;\;\; & \tilde{\mu}+\tilde{\gamma}EZ\geq k\\
\end{split}
\end{equation}
with
\begin{equation}
g(\tilde{\mu}, \tilde{\gamma})=(\tilde{\mu}, \tilde{\gamma}, 1)D^{-1} (\tilde{\mu}, \tilde{\gamma}, 1)^T,
\end{equation}
and
\begin{equation}D=
\begin{pmatrix}
\mu^T\Sigma^{-1}\mu&\gamma^T\Sigma^{-1}\mu&e^T\Sigma^{-1}\mu\\
\mu^T\Sigma^{-1}\gamma &\gamma^T\Sigma^{-1}\gamma &e^T\Sigma^{-1}\gamma \\
\mu^T\Sigma^{-1}e&\gamma^T\Sigma^{-1}e&e^T\Sigma^{-1}e\\
 \end{pmatrix}.
\end{equation}
While this approach gives an expression for the optimal portfolio as in (\ref{16}), one still needs to solve another optimization problem, namely (\ref{17}).

In the following, we will show that (\ref{16}) and (\ref{17}) can be simplified further. We first fix some notation. Similar to the case of $\gamma$ in Section 2, we write $\mu$ as a linear combination of $A_1^c, A_2^c, \cdots, A_n^c$ as $\mu=\sum_{i=1}^n\mu_iA_i^c$. We denote by $\mu_0=(\mu_1, \mu_2, \cdots, \mu_n)^T$ the column vector formed by the corresponding coefficients of this linear combination. Also, for each NMVM  distribution $X$ as in (\ref{2ndx}), we define
\begin{equation}\label{y}
Y\overset{d}=\mu_0+\gamma_0Z+\sqrt{Z}N_n.
\end{equation}
and we call $Y$ the NMVM distribution associated with $X$. In the $x$-coordinate system, problem (\ref{15}) can be written as
\begin{equation}\label{op2a2}
\begin{split}
&\min_{x\in \mathcal{R}_D}\rho(x^TY)\\
\mbox{s.t.}\;\;\; &x^Te_A=1\\
 &x^Tm\geq k.
\end{split}
\end{equation}
where $e_A=A^{-1}e$ and now $m=\mu_0+\gamma_0 EZ$.


The solution of (\ref{op2a2}) can be characterized by using the same idea as in Proposition 3.1 of \cite{Shi_And_Kim_2021}. Below, we  write down the optimal solution of (\ref{op2a2}) as a corollary of their result.

\begin{corollary} For any law invariant coherent risk measure $\rho$, the solution of (\ref{op2a2}) is given by
\begin{equation}\label{22}
x^{\star}=(\mu_0, \gamma_0, e_A)G^{-1}(\tilde{\mu}_0^{\star}, \tilde{\gamma}_0^{\star}, 1)^T,
\end{equation}
where
\begin{equation}\label{k}
\begin{split}
(\tilde{\mu}_0^{\star}, \tilde{\gamma}_0^{\star})=&\mbox{argmin}_{\tilde{\mu}_0, \tilde{\gamma}_0}\rho(\tilde{\mu}_0+\tilde{\gamma}_0Z+\sqrt{g(
\tilde{\mu}_0, \tilde{\gamma}_0)Z}N),\\
&\\
\mbox{s.t.}\;\;\; & \tilde{\mu}_0+\tilde{\gamma}_0EZ\geq k,\\
\end{split}
\end{equation}
with
\begin{equation}
g(\tilde{\mu}_0, \tilde{\gamma}_0)=(\tilde{\mu}_0, \tilde{\gamma}_0, 1)G^{-1} (\tilde{\mu}_0, \tilde{\gamma}_0, 1)^T,
\end{equation}
and
\begin{equation}G=
\begin{pmatrix}
\mu^T_0\mu_0&\gamma^T_0\mu_0&e^T_A\mu_0\\
\mu^T_0\gamma_0 &\gamma^T_0\gamma_0 &e^T_A\gamma_0 \\
\mu^T_0e_A&\gamma^T_0e_A&e^T_Ae_A\\
 \end{pmatrix}.
\end{equation}
\end{corollary}
\begin{proof}
The proof follows from the same argument as in the proof of Proposition 3.1 of \cite{Shi_And_Kim_2021}. In our case here, we have $\Sigma=I$, and with this, their formula reduces to the solution in the corollary.
\end{proof}

We remark that the optimal solutions $\omega^{\star}$ in (\ref{16}) and $x^{\star}$ in (\ref{22}) are related by $x^{\star}=(\omega^{\star})^TA$.




\subsection{Closed form approximations for risks}

The result (\ref{16}) reduces the high dimensional nature of the portfolio optimization problem  (\ref{15}) to two dimensions, as claimed in \cite{Shi_And_Kim_2021}.
In particular, this approach reduces the computational time of determining the optimal portfolio significantly, as claimed in their paper. In this section, we take an alternative approach and attempt to reduce the computational time of such optimization problems also. For this, we linearly transform the portfolio space into a different coordinate system.


The following simple lemma shows that to determine the optimal portfolio, one only needs to solve a low-dimensional optimization problem in, as explained in Remark \ref{rem1} below.



\begin{lemma} \label{lem2} For any coherent risk measure $\rho$, we have
\begin{equation}\label{26}
\rho(x^TY) =-x^T\mu_0+\rho(aZ+\sqrt{Z}N)\sqrt{x^Tx},
\end{equation}
where
\[
a=||\gamma_0||cos [\theta(\gamma_0, x)],
\]
and $\theta(\gamma_0, x)$ is the angle between $\gamma_0$ and $x$.
\end{lemma}
\begin{proof} We have $x^TY\overset{d}{=}x^T\mu_0+x^T\gamma_0Z+||x||\sqrt{Z}N$. Since $\rho$ is law invariant and cash invariant, we have $\rho(x^TY)=-x^T\mu_0+\rho(x^T\gamma_0Z+||x||\sqrt{Z}N)$. Now, since $x^T\gamma_0=||x||||\gamma_0||cos[\theta(x,\gamma_0)]$ and $\rho$ is positive homogeneous, (\ref{26}) follows.
\end{proof}




\begin{remark} \label{rem1} We can also write  $x^T\mu_0=||x||||\mu_0||Cos\theta[(\mu_0, x)]$ in (\ref{26}). This shows that optimization problems like $\min_{\mathcal{R}_D}\rho(x^TY)$ on the image $\mathcal{R}_D$ of some portfolio set $D$ under the transformation defined in (\ref{T}) depend on the norm $||x||$ of $x$, the angle between $x$ and $\mu_0$, and the angle between $x$ and $\gamma_0$.
\end{remark}



Lemma (\ref{26}) shows that the risk $\rho(x^TY)$ is dependent on $\rho(aZ+\sqrt{Z}N)$. From now on we write
\begin{equation}\label{hhh}
h(a)=:\rho(aZ+\sqrt{Z}N), \; \; \; \mbox{for} \;\; a=||\gamma_0||cos[\theta(x, \gamma_0)].
\end{equation}
In the following lemma, we state some properties of $h(a)$.

\begin{lemma}\label{monotone} For any coherent risk measure $\rho$, the function
$h(a)=\rho(aZ+\sqrt{Z}N)$ is a decreasing, convex, and continuous function of $a$ in $(-||\gamma_0||, ||\gamma_0||)$.
\end{lemma}
\begin{proof}
The decreasing property of $h(a)$ follows from Theorem 3.1 of \cite{Shi_And_Kim_2021}. Here, for the sake of being self-contained, we present its proof. Take two $a_2\geq a_1$. Then the subadditivity and monotonicity of $\rho$ implies
\begin{equation}
\begin{split}
h(a_2)=\rho(a_2Z+\sqrt{Z}N)=&\rho(a_1Z+(a_2-a_1)Z+\sqrt{Z}N) \\
\le &\rho(a_1Z+\sqrt{Z}N)+\rho((a_2-a_1)Z)\\
\le &\rho(a_1Z+\sqrt{Z}N)+\rho(0)\\
=&\rho(a_1Z+\sqrt{Z}N)=h(a_1).
\end{split}
\end{equation}
The convexity of $h(a)$ follows from
\begin{equation}
\begin{split}
 h(\lambda a+(1-\lambda)b)=&\rho(\lambda [aZ+\sqrt{Z}N]+(1-\lambda)[bZ+\sqrt{Z}N] ) \\
 \le& \rho(\lambda [aZ+\sqrt{Z}N])+\rho((1-\lambda)[bZ+\sqrt{Z}N] )\\
 =& \lambda h(a)+(1-\lambda)h(b),
 \end{split}
\end{equation}
for any $1\geq \lambda \geq 0$. Since $h(a)$ is a convex function, it is continuous in the interior of its domain, which is $(-||\gamma_0||, ||\gamma_0||)$.
\end{proof}

\begin{remark} Note that from Lemma \ref{lem2}, we have $\rho(x^TY)=-x^T\mu_0+\sqrt{x^Tx}h(a)$.  Since $h(a)$ is a continuous function in  $(-||\gamma_0||, ||\gamma_0||)$, as was shown in Lemma \ref{monotone}, we have
\[
h(a)\approx \sum_{i=1}^{n-1}1_{[a_i, a_{i+1})}(a)\rho(a_iZ+\sqrt{Z}N)
\]
for an appropriately chosen partition $a_1\le a_2\le \cdots \le a_n$ of the interval $(-||\gamma_0||, ||\gamma_0||)$. Therefore we have the approximation
\begin{equation}\label{approx}
\rho(x^TY)\approx -x^T\mu_0+\sqrt{x^Tx}\sum_{i=1}^{n-1}1_{[a_i, a_{i+1})}(a)\rho(a_iZ+\sqrt{Z}N)
\end{equation}
with $a=-||\gamma_0||cos[\theta(\gamma_0, x)]$ when the mesh
$\max_{1\le i\le n-1}(a_{i+1}-a_{i})$ of the partition
 \[
 -||\gamma_0||=a_1\le a_2\le \cdots \le a_n=||\gamma_0||.
\]
is sufficiently small.
\end{remark}

\begin{remark} We remark that in order to obtain high precision, the points $a_1, a_2, \cdots, a_n$ need to be chosen to make $\max_i(a_i-a_{i-1})$ as small as possible. However, this comes with a cost, as we need to calculate $\rho(a_iZ+\sqrt{Z}N)$ for each $a_i$ to obtain an approximate value of $\rho(x^TY)$.
\end{remark}

This proposition gives an accurate approximation for the value of the risk measure when the mesh of the corresponding partition is very small. However, as stated earlier, we need to calculate $h(a_i)=\rho(a_iZ+\sqrt{Z}N)$ for many $i$.

When the value of $||\gamma_0||$ is relatively small, we can get an even simpler approximation, as stated in the following proposition. The next proposition gives an approximation for the risk measure in which one needs to calculate $h(a)$ for only two values of $a$ to obtain a good approximation. In the following, we use the following notation 
$g=:\sqrt{\sum_{i=1}^d\gamma_i^2}$.

\begin{proposition} \label{3and7} Any law invariant coherent risk measure $\rho$ can be approximated by
\begin{equation}\label{31}
\bar{\rho} (\omega^TX)=: -x^T\mu_0+\sqrt{x^Tx}\Big [\frac{1-Cos[\theta(x, \gamma_0)]}{2}h(-g)+ \frac{1+Cos[\theta(x, \gamma_0)]}{2}h(g)\Big ],
\end{equation}
where $\theta[x, \gamma_0]$ is the angle between $x^T=\omega^TA$ and $\gamma_0$.
\end{proposition}
\begin{proof} We have $\rho(x^TY)=-x^T\mu_0+\sqrt{x^Tx}h(a)$ from Lemma \ref{lem2}. We approximate $h(a)$ by the line that passes through $(-g, h(-g))$ and $(g, h(g))$. The equation of this line is
\[
y(a)-h(-g)=-\frac{h(-g)-h(g)}{2g}(a+g)=\frac{h(-g)+h(g)}{2}-
\frac{h(-g)+h(g)}{2}cos[\theta(x, \gamma_0)].
\]
Replacing $h(a)$ by $y(a)$ in the expression for $\rho(x^TY)$ above and combining the terms $h(-g)$ and $h(g)$ gives (\ref{31}).
\end{proof}
\begin{remark} The precision of the approximation (\ref{31}) clearly depends on the properties of the risk measure $\rho$. More specifically, the properties of $h(a)$ in (\ref{hhh}) determine the degree of accuracy of our approximation. The only information on $h(a)$ that we know is its continuity, convexity, and the decreasing property as stated in Lemma \ref{monotone}. These are not sufficient to study the  accuracy of the approximation for this proposition. However, our extensive numerical tests show that when $g$ is relatively small, which is usually the case for EM estimates of the NMVM models from empirical data, these approximations work pretty well.
\end{remark}

The above proposition simplifies the computations of optimal portfolios considerably, as one only needs to evaluate $h(a)$ at the two points $-||\gamma_0||$ and $-||\gamma_0||$.
\begin{remark} If $\gamma=0$, then $\gamma_0=0$ and the right-hand side of (\ref{31}) reduces to
\begin{equation}\label{4a8}
-x^T\mu_0+\sqrt{x^Tx}\rho(\sqrt{Z}N).
\end{equation}
The expression (\ref{4a8}) is exactly the formula for risk measures for elliptical distributions, as discussed in \cite{Landsman_And_Valdez_2003}.
\end{remark}

\begin{remark} We remark that in optimization problems like $\min_{\omega}\rho(\omega^TX)$, one can  optimize $\bar{\rho}(\omega^TX)$ instead and obtain an approximately optimal portfolio. Observe that $\bar{\rho}(x^TY)$ can also be written as
\begin{equation}\label{3a5}
  \bar{\rho}(x^TY)=-x^T\mu_0- \frac{h(-g)-h(g)}{2||\gamma_0||}x^T\gamma_0+\frac{h(g)+h(-g)}{2}\sqrt{x^Tx}.
\end{equation}
Here, $\frac{h(-g)-h(g)}{2||\gamma_0||}$ is a positive number, as $h(a)$ is a decreasing function, as stated in Lemma 2.8. The minimization and maximization of functions of the form (\ref{3a5}) were discussed in  \cite{Landsman_2008} in detail, see also \cite{Owadally_2011} and \cite{Markov_2016}.
\end{remark}

\begin{remark} We should mention that (\ref{31}) was obtained under the assumption that the portfolio space is the whole of $R^n$. For problems associated with some small portfolio space, the value of $||\gamma_0||cos[\theta(x, \gamma_0)]$ does not cover the whole interval $[-||\gamma_0||, -||\gamma_0||]$. Therefore, in such cases, (\ref{31}) needs to be adjusted appropriately.
\end{remark}




As mentioned earlier, calculating the VaR and CVaR needs numerical procedures or Monte Carlo approaches for most models of asset returns. In the past, calculations of the VaR relied on linear approximations of the portfolio risks, see, e.g., \cite{Duffie_And_Pan_1997} and  \cite{Jorion_1996}),  or Monte Carlo simulation-based tools, see, e.g., \cite{Uryasev_Stanislav_2000}, \cite{Bucay_And_Rosen_1999}, and \cite{Pritsker_1997}. Next, we present closed form approximations for them, as a corollary to Proposition \ref{3and7}.


Acerbi and Tasche (2010) defines $VaR_{\alpha}(X)$ to be the negative of the upper quantile
\[
q^{\alpha}(X)=\inf\{x\in R: P(X\le x)>\alpha\}
\]
of $X$, i.e., $VaR_{\alpha}(X)=-q^{\alpha}(X)$. When the random variable $X$ has a positive probability density function, we have $q^{\alpha}(X)=F_{X}^{-1}(\alpha)$ and so $VaR_{\alpha}(X)=-F_X^{-1}(\alpha)$. With this definition, VaR is a positive homogeneous monetary risk measure. The conditional value at risk is defined by    $CVaR_{\alpha}(X)=-E[X|X\le -VaR_{\alpha}]$. Then, for any portfolio $\omega$, the CVaR of the loss $-\omega^TX$ and
return $\omega^TX$ are given by
\begin{equation}\label{HK1}
CVaR_{\beta}(-\omega^TX)=\frac{1}{\beta}\int_{F_{\omega}^{-1}(1-\beta)}^{+\infty}yf_{\omega}(y)dy, \; \; CVaR_{\beta}(\omega^TX)=-\frac{1}{\beta}\int_{-\infty}^{F_{\omega}^{-1}(\beta)}yf_{\omega}(y)dy
\end{equation}
respectively. When $X\sim N(\mu, \sigma^2)$, a straightforward calculation shows that these risk measures have the closed form expressions
\begin{eqnarray}
VaR_{\alpha}(X)=-\mu-\sigma \Phi^{-1}(\alpha), \;\; CVaR_{\alpha}(X)=-\mu+\sigma \frac{e^{-[\Phi^{-1}(\alpha)]^2/2}}{\alpha \sqrt{2\pi}},
\end{eqnarray}
where $\Phi(\cdot)$ is the cumulative distribution function of the standard normal random variable. When $X$ is elliptically distributed, one can also express these risk measures in closed form, as in equation (2) of Landsman and Valdez (2000).

Below, we discuss this risk measure for the class $Y_a=:aZ+\sqrt{Z}N$ of random variables as this is sufficient to obtain expressions for these risk measures for general NMVM models, due to (\ref{26}).

Denoting the probability density function of $Z$ by $g(s)$, the probability density function of $Y_a$ is given by
\begin{equation}
g_a(y)=\frac{1}{\sqrt{2\pi}}\int_0^{+\infty}\frac{g(s)}{\sqrt{s}}e^{-\frac{(y-as)^2}{2s}}ds.
\end{equation}
If $Z\sim GIG(\lambda, \chi, \psi)$, then $Y_{a}$ has the following density function
\begin{equation}\label{agig}
f_a(y)=\frac{(\sqrt{\psi/\chi})^{\lambda} (\psi+a^2)^{\frac{1}{2}-\lambda}}{\sqrt{2\pi}K_{\lambda}(\sqrt{\chi \psi})}\times  \frac{K_{\lambda-\frac{1}{2}}(\sqrt{(\chi+y^2)(\psi+a^2)})e^{ay}}{(\sqrt{(\chi+y^2)(\psi+a^2)})^{\frac{1}{2}-\lambda}},
\end{equation}
where $K_{\lambda}(\cdot)$ is the modified Bessel function of the third kind. Write $y_{\beta}(a)=:VaR_{\beta}(Y_a)$ and note that $-y_{\beta}(a)$ is the $\beta$ quantile of $Y_a$. By using the definitions of VaR and CVaR, we obtain the following relations:
\begin{lemma} We have
\begin{equation}\label{beta}
    \beta=\int_0^{+\infty}g(s)\Phi(\frac{-y_{\beta}(a)-as}{\sqrt{s}})ds,
\end{equation}
and
\begin{equation}\label{beta2}
CVaR_{\beta}(Y_a)=-\frac{1}{\beta}\int_0^{+\infty}\Big [as\Phi(\frac{-y_{\beta}(a)-as}{\sqrt{s}})-\sqrt{\frac{s}{2\pi}}e^{-\frac{(y_{\beta}(a)+as)^2}{2s}}\Big ]g(s)ds,
\end{equation}
where $\Phi$ is the  cumulative distribution  function of the standard normal random variable.
\end{lemma}
\begin{proof} By the definition of $y_{\beta}(a)$, we have
$P(Y_a\le -y_{\beta}(a))=\beta$. We have $P(Y_a\le -y_{\beta}(a))=\int_0^{+\infty}P(N\le\frac{-y_{\beta}(a)-as}{\sqrt{s}})g(s)ds=\int_0^{+\infty}\Phi(\frac{-y_{\beta}(a)-as}{\sqrt{s}})g(s)ds$. This completes the proof of (\ref{beta}). To obtain (\ref{beta2}), note that $CVaR_{\beta}(Y_a)=-E[Y_a/Y_a\le -y_{\beta}(a)]=-\frac{E[Y_a1_{\{Y_a\le -y_{\beta}(a)\}}]}{P(Y_{a}\le -y_{\beta}(a))}$. We have $P(Y_a\le -y_{\beta}(a))=\beta$ by the definition of $y_{\beta}(a)$. The numerator equals
\[
E[Y_a1_{\{Y_a\le -y_{\beta}(a)\}}]=\int_0^{+\infty}E[N(as, s)1_{\{N(as, s)\le -y_{\beta}(a)\}}]g(s)ds.
\]
We can easily calculate $E[N(as, s)1_{\{N(as, s)\le -y_{\beta}(a)\}}]=as\Phi(\frac{-y_{\beta}(a)-as}{\sqrt{s}})-\sqrt{
\frac{s}{2\pi}}e^{-\frac{(y_{\beta}(a)+as)^2}{2s}}$ and this completes the proof.
\end{proof}

Now, if we apply (\ref{26}) to the risk measures VaR and CVaR, we obtain
\begin{equation}\label{beta1}
\begin{split}
VaR_{\beta}(\omega^TX)=Var_{\beta}(x^TY)=&-x^T\mu_0+VaR_{\beta}(Y_a)\sqrt{x^Tx}\\
=&-x^T\mu_0+y_{\beta}(a)\sqrt{x^Tx},
\end{split}
\end{equation}
and
\begin{equation}\label{beta3}
CVaR_{\beta}(\omega^TX)=CVaR_{\beta}(x^TY)=-x^T\mu_0+CVaR_{\beta}(Y_a)\sqrt{x^Tx},
\end{equation}
where $x^T=\omega^TA$, $a=||\gamma_0||cos(x, \gamma_0)$, and  $y_{\beta}$ satisfies (\ref{beta}). Therefore, optimization problems like $\min_DCVaR_{\beta}(\omega^TX)$ or $\min_DVaR_{\beta}(\omega^TX)$ for some domain $D$ of portfolios involve computing (\ref{beta}) or (\ref{beta2}) for each $x^T=\omega^TA$.

Below, we apply Proposition 2.8 and obtain simpler expressions for the VaR and CVaR.

\begin{theorem}\label{312} The $VaR_{\beta}(\cdot)$ and $CVaR_{\beta}(\cdot)$ can be approximated by the following $V_{\beta}(\cdot)$ and $CV_{\beta}(\cdot)$ respectively
\begin{equation}\label{mm}
V_{\beta}(\omega^TX)=: -x^T\mu_0+\sqrt{x^Tx}\Big [w_{+}+w_{-}Cos[\theta(x, \gamma_0)]\Big ],
\end{equation}
and
\begin{equation}\label{mmm}
CV_{\beta}(\omega^TX)=:-x^T\mu_0+\sqrt{x^Tx}\Big [v_{+}+v_{-}Cos[\theta(x, \gamma_0)]\Big ],
\end{equation}
where $x=\omega^TA$ and the constants $w_{+}, w_{-}, v_{+}, v_{-},$ are given by
\begin{equation}\label{4a6}
\begin{split}
w_{+}=&\frac{VaR_{\beta}(Y_b)+VaR_{\beta}(Y_{-b})}{2},\;\; v_{+}=\frac{CVaR_{\beta}(Y_b)+CVaR_{\beta}(Y_{-b})}{2},\\
w_{-}=& \frac{VaR_{\beta}(Y_b)-VaR_{\beta}(Y_{-b})}{2},\; v_{-}=\frac{CVaR_{\beta}(Y_b)-CVaR_{\beta}(Y_{-b})}{2},
\end{split}
\end{equation}
where $Y_b=bZ+\sqrt{Z}N$ and $Y_{-b}=-bZ+\sqrt{Z}N$ with $b=||\gamma_0||$.
\end{theorem}
\begin{proof}
The proof follows from Proposition (\ref{31}).
\end{proof}

We remark that optimization problems like
$\min_{R^n}CVaR_{\beta}(\omega^TX)$ or $\min_{R^n}VaR_{\beta}(\omega^TX)$ can be replaced by $\min_{R^n}CV_{\beta}(\omega^TX)$ and  $\min_{R^n}V_{\beta}(\omega^TX)$ respectively by using the expressions in (\ref{mm}) and (\ref{mmm}). The latter ones are computationally much efficient as they involve the values of the VaR and CVaR of $Y_b=bZ+\sqrt{Z}N$ and $Y_{-b}=-bZ+\sqrt{Z}N$ only for $b=||\gamma_0||$.



\section{Numerical results}
In this section, we numerically check the performance of our results. First, we fit the GH distribution to empirical data of five stocks by using the EM algorithm. For this, we use five years of price history of the stocks AMD, CZR, ENPH, NVDA,  and TSLA from the  2nd of January 2015 to the 30th of December 2020. The following table gives a summary of this data
\begin{table}[h]
\centering
\caption{Data Describe}
\begin{tabular}{cllllll}
\hline
Name & AMD       & CZR       & ENPH      & NVDAC      & TSLA      \\
\hline
\hline
mean & 0.003121  & 0.002738  & 0.003218  & 0.002568  & 0.002432  \\
std  & 0.040087  & 0.04008   & 0.056228  & 0.028707  & 0.034737  \\
min  & -0.242291 & -0.37505  & -0.373656 & -0.187559 & -0.210628 \\
max  & 0.522901  & 0.441571  & 0.424446  & 0.298067  & 0.198949  \\ \hline
\end{tabular}
\end{table}

We apply the modified EM scheme to fit the daily log-returns of these stocks to  5-dimensional GH distributions. The algorithm is called multi-cycle, expectation, conditional estimation (MCECM) procedure, see \cite{McNeil_Alexander_J_And_Frey_Rudigger_And_Embrechts_Paul_2015} \cite{Meng_Xiao-Li_And_Rubin_Donald_B_1993} for the details of this algorithm. First, we fit the return data to the model $\gamma Z+\sqrt{Z}AN_n$. In this case, the estimated parameters for $Z$
are
\begin{gather*}
    \lambda = -\frac{1}{2},\\
    \chi = 0.87953198,\\
    \psi = 0.645169932,\\
\end{gather*}
and
\begin{equation*}
\begin{split}
  \gamma =& \left[ 0.00268318,	0.00147543,	0.00273905,	0.00145453,	0.00180711\right],\\
    \Sigma =& \begin{bmatrix}
& 0.001341 & 0.000253 & 0.000398 & 0.000529 & 0.000333 \\
 & 0.000253 & 0.001034 & 0.0003   & 0.00025  & 0.000269 \\
 & 0.000398 & 0.0003   & 0.00285  & 0.000274 & 0.000321 \\
 & 0.000529 & 0.00025  & 0.000274 & 0.000675 & 0.000311 \\
 & 0.000333 & 0.000269 & 0.000321 & 0.000311 & 0.00109  \\
    \end{bmatrix}.
\end{split}
\end{equation*}

With these parameters, we have $EZ=\frac{\chi}{\psi}, Var(Z)=\frac{\chi}{\psi^3}, m_3(Z)=\frac{3\chi}{\psi^5}$. We also have $m=\gamma_0EZ=\frac{\chi}{\psi} \gamma_0, \; e_A=\Sigma^{-\frac{1}{2}}e$.
As in \cite{Zhao_Shangmei_And_Lu_Qing_And_Han_Liyan_And_Liu_Yong_And_Hu_Fei_2015}, we consider returns  $r\in [0, 0.2]$. We take $r_i=\frac{0.2}{100}i, i=0,1,2,\cdots, 100$. For each $r_i$, we calculate $x^i, i=1,2, \cdots, 20,$ from $x^i=s_im+t_ie_A$, where $s_i$ and $t_i$ are given by (\ref{mmm}). Then for each $i$, we calculate $\mbox{SKew}((x^i)^TY)$ from (\ref{1a5}), $CVaR_{\beta}((x^i)^TY)$ from (\ref{HK1}) and we plot $(\mbox{SKew}((x^i)^TY), r_i, CVaR_{\beta}((x^i)^TY))$ in three-dimensional space. Table \ref{Optimal weights of the 5 stocks and the skewness} gives a summary of the optimal portfolios and the corresponding skewness. Figure \ref{Six_Stocks} is the mean-CVaR-skewness
efficient frontier.

\begin{figure}[H]
    \centering
    \includegraphics[scale=0.7]{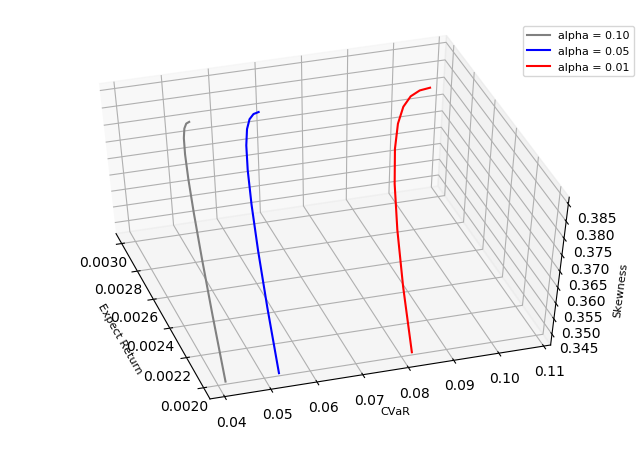}
    \caption{The mean-CVaR-skewness frontier}
    \label{Six_Stocks}
\end{figure}

\begin{table}[h]
\centering
\caption{Optimal weights of the 5 stocks and the skewness}
\label{Optimal weights of the 5 stocks and the skewness}
\begin{tabular}{llllll}
\hline
\hline
Expect Return & 0.0020   & 0.0022   & 0.0024   & 0.0027   & 0.0029   \\ \hline
\hline
$\omega_1$           & 0.077077 & 0.194069 & 0.31106  & 0.428051 & 0.545042 \\
$\omega_2$             & 0.252863 & 0.22433  & 0.195798 & 0.167265 & 0.138732 \\
$\omega_3$            & 0.067729 & 0.101723 & 0.135716 & 0.169709 & 0.203703 \\
$\omega_4$            & 0.399764 & 0.26734  & 0.134915 & 0.00249  & -0.12994 \\
$\omega_5$       & 0.202566 & 0.212539 & 0.222512 & 0.232485 & 0.242458 \\
Skewness       & 0.34231  & 0.370487 & 0.383957 & 0.385706 & 0.380047 \\
\hline
\end{tabular}
\end{table}

In the second case, we fit the return data of the above listed five stocks
to the model
\[
X=\mu+\gamma Z+\sqrt{Z}N_n.
\]
The following table lists the numerical values for the parameters of our fit.

\begin{gather*}
    \lambda = -0.378655004,\\
    \chi = 0.379275063,\\
    \psi = 0.371543387,\\
\end{gather*}
The values of $\gamma$ and $\Sigma$ are as follows.
\begin{equation*}
\centering
\label{Params_GH}
\begin{split}
  \mu =& \left[ 0.00041332,	0.00152207,	0.00058012,	0.00156685,	0.0006603\right],\\
  \gamma = & \left[ 0.00163631,	0.00073499,	0.00159418,	0.000605,	0.00107086 \right],\\
    \Sigma =& \begin{bmatrix}
 & 0.001341 & 0.000253 & 0.000398 & 0.000529 & 0.000333 \\
 & 0.000253 & 0.001034 & 0.0003   & 0.00025  & 0.000269 \\
 & 0.000398 & 0.0003   & 0.00285  & 0.000274 & 0.000321 \\
 & 0.000529 & 0.00025  & 0.000274 & 0.000675 & 0.000311 \\
 & 0.000333 & 0.000269 & 0.000321 & 0.000311 & 0.00109  \\
    \end{bmatrix}.
\end{split}
\end{equation*}
By using the definitions of $\mu_0$ and $\gamma_0$, we calculate them as follows.
\begin{equation*}
\centering
\begin{split}
  \mu_0 =& \left[-0.0064934,	0.0357838,	0.0046229,	0.0524977,	0.0067647\right],\\
  \gamma_0 = & \left[ 0.0339663, 0.0125813,	0.0211417,	0.0018441,	0.0207596 \right].\\
\end{split}
\end{equation*}

Table \ref{VaR_3.12} tests  the  performance  of  our  Theorem  3.12.   The  value  of  the VaR  is  calculated  by a numerical root finding method and also by using Theorem 3.12.  It can be seen that the results match very well.

\begin{table}[h]
\caption{Numerical VaR vs. Theorem 3.12}
\centering
\label{VaR_3.12}
\begin{tabular}{l|lll|lll}
\hline
\multicolumn{1}{c|}{Portfolio Weights} & \multicolumn{3}{c|}{VaR}     & \multicolumn{3}{c}{V}       \\
\multicolumn{1}{c|}{}                                  & {$\beta$} = 0.1     & {$\beta$} = 0.05    & {$\beta$} = 0.01    & {$\beta$} = 0.1     & {$\beta$} = 0.05    & {$\beta$} = 0.01    \\ \hline \hline
{[}0.1, 0.4, 0.2, 0.1, 0.2{]}         & 0.027729      & 0.042165      & 0.082055       & 0.029828      & 0.044277      & 0.084206       \\
{[}0.2, 0.1, 0.5, 0.1, 0.1{]}         & 0.038508      & 0.058196      & 0.112613       & 0.040006      & 0.059714      & 0.114193       \\
{[}0.1, 0.4, 0.1, 0.3, 0.1{]}         & 0.02568       & 0.039183      & 0.076513       & 0.02816       & 0.041678      & 0.079054       \\
{[}0.3, 0.1, 0.3, 0.1, 0.2{]}         & 0.03131       & 0.047384      & 0.091771       & 0.032775      & 0.048857      & 0.093264       \\
{[}0.1, 0.3, 0.1, 0.3, 0.2{]}         & 0.025091      & 0.038256      & 0.074639       & 0.027397      & 0.040574      & 0.076994       \\ \hline
\end{tabular}

\end{table}
\newpage

In Table \ref{CVaR_3.12} we compare the numerical calculation of the CVaR with the performance of the approximation of the CVaR in Theorem 3.12.  Again, it can be seen that the results match very well.

\begin{table}[h]
\centering
\caption{Numerical CVaR vs. Theorem 3.12}
\label{CVaR_3.12}
\begin{tabular}{l|lll|lll}
\hline
\multicolumn{1}{c|}{Portfolio Weights}     & \multicolumn{3}{c|}{CVaR}    & \multicolumn{3}{c}{CV}      \\
\multicolumn{1}{r|}{}    & {$\beta$} = 0.1     & {$\beta$} = 0.05    & {$\beta$} = 0.01    & {$\beta$} = 0.1     & {$\beta$} = 0.05    & {$\beta$} = 0.01    \\ \hline \hline
{[}0.1, 0.4, 0.2, 0.1, 0.2{]}         & 0.050643      & 0.067315      & 0.111296       & 0.050654      & 0.067341      & 0.111366       \\
{[}0.2, 0.1, 0.5, 0.1, 0.1{]}         & 0.069764      & 0.092505      & 0.152508       & 0.0698        & 0.092567      & 0.15264        \\
{[}0.1, 0.4, 0.1, 0.3, 0.1{]}         & 0.04712       & 0.062719      & 0.103883       & 0.047136      & 0.062755      & 0.103972       \\
{[}0.3, 0.1, 0.3, 0.1, 0.2{]}         & 0.056816      & 0.075369      & 0.124296       & 0.056815      & 0.075376      & 0.124327       \\
{[}0.1, 0.3, 0.1, 0.3, 0.2{]}         & 0.04599       & 0.061195      & 0.10131        & 0.046         & 0.06122       & 0.101378       \\ \hline
\end{tabular}
\end{table}


\section{Conclusion}
\cite{Zhao_Shangmei_And_Lu_Qing_And_Han_Liyan_And_Liu_Yong_And_Hu_Fei_2015} showed that mean-CVaR-skewness portfolio optimization problems based on asymmetric Laplace distributions can be transformed into  quadratic optimization problems. In this note, we extended their result and showed that mean-risk-skewness portfolio optimization problems based on a larger class of NMVM models can also be transformed into quadratic optimization problems under any law invariant risk measure.
The critical step to achieve this was to  transform the original portfolio space into another space by an appropriate linear transformation, a step which  enabled us to express both the risk and skewness as functions of a single  variable. By showing that any law invariant coherent risk measure is an increasing function and skewess is a decreasing function of this variable, we were able to transform the original optimization problem into a quadratic optimization problem as in  \cite{Zhao_Shangmei_And_Lu_Qing_And_Han_Liyan_And_Liu_Yong_And_Hu_Fei_2015}. In the rest of this paper, we made use of this transformation to come up with approximate closed form formulas for law invariant risk measures and hence also for the VaR and CVaR. Our numerical tests show that such closed form approximations are accurate.
\bibliographystyle{MybibMeanSkew}
\bibliography{main}

\begin{thebibliography}{50}
\providecommand{\natexlab}[1]{#1}

\bibitem[{Aas \& Haff(2006)}]{Aas_Kjersti_And_Haff_Ingrid_Hobaek_2006}
Aas, K. \& Haff, I. H. (2006).
\newblock The generalized hyperbolic skew {S}tudent’s $t$-distribution.
\newblock \emph{Journal of Financial Econometrics}, 4, 275--309

\bibitem[{Acerbi \& Tasche(2002)}]{Acerbi_Carlo_And_Tasche_Dirk_2002}
Acerbi, C. \& Tasche, D. (2002).
\newblock On the coherence of expected shortfall.
\newblock \emph{Journal of Banking \& Finance}, 26, 1487--1503

\bibitem[{Akturk \& Ararat(2020)}]{Ararat_2020}
Akturk, T. D. \& Ararat, C. (2020).
\newblock Portfolio optimization with two coherent risk measures.
\newblock \emph{Journal of Global Optimization}, 597–626

\bibitem[{Artzner \emph{et~al.}(1999)Artzner, Delbaen, Eber, \&
  Heath}]{Artzner_PhilippeAndDelbaen_FreddyAndEber_Jean-MarcAndHeath_David_1999}
Artzner, P., Delbaen, F., Eber, J. M., \& Heath, D. (1999).
\newblock Coherent measures of risk.
\newblock \emph{Mathematical Finance}, 9, 203--228

\bibitem[{Bali \& Cakici(2004)}]{Bali_And_Cakici_2004}
Bali, T. G. \& Cakici, N. (2004).
\newblock Value at risk and expected stock returns.
\newblock \emph{Financial Analysts Journal}, 60, 57--73

\bibitem[{Barndorff-Nielsen(1997)}]{Barndorff-Nielsen_Ole_E_1997}
Barndorff-Nielsen, O. E. (1997).
\newblock Processes of normal inverse {G}aussian type.
\newblock \emph{Finance and Stochastics}, 2, 41--68

\bibitem[{Bingham \&
  Kiesel(2001)}]{Bingham_NICHOLAS_H_And_Kiesel_Rudigger_2001}
Bingham, N. H. \& Kiesel, R. (2001).
\newblock Modelling asset returns with hyperbolic distributions.
\newblock Return Distributions in Finance, 1--20. Elsevier

\bibitem[{Bollerslev \& Todorov(2011)}]{Bollerslev_And_Todorov_2011}
Bollerslev, T. \& Todorov, V. (2011).
\newblock Tails, fears, and risk premia.
\newblock \emph{The Journal of Finance}, 66, 2165--2211

\bibitem[{Bucay \& Rosen(1999)}]{Bucay_And_Rosen_1999}
Bucay, N. \& Rosen, D. (1999).
\newblock Credit risk of an international bond portfolio: A case study.
\newblock \emph{ALGO Research Quarterly}, 2, 9--29

\bibitem[{Chekhlov \emph{et~al.}(2005)Chekhlov, Uryasev, \&
  Zabarankin}]{Chekhlov_And_Uryasev_And_Zabarankin_2005}
Chekhlov, A., Uryasev, S., \& Zabarankin, M. (2005).
\newblock Drawdown measure in portfolio optimization.
\newblock \emph{International Journal of Theoretical and Applied Finance}, 8,
  13--58

\bibitem[{Cont \& Tankov(2004)}]{Cont_Rama_And_Tankov_Peter_2004}
Cont, R. \& Tankov, P. (2004).
\newblock Nonparametric calibration of jump-diffusion option pricing models.
\newblock \emph{The Journal of Computational Finance}, 7, 1--49

\bibitem[{Dana(2005)}]{Dana_Rose-Anne_2005}
Dana, R. A. (2005).
\newblock A representation result for concave {S}chur concave functions.
\newblock \emph{Mathematical Finance: An International Journal of Mathematics,
  Statistics and Financial Economics}, 15, 613--634

\bibitem[{Duffie \& Pan(1997)}]{Duffie_And_Pan_1997}
Duffie, D. \& Pan, J. (1997).
\newblock An overview of value at risk.
\newblock \emph{Journal of Derivatives}, 4, 7--49

\bibitem[{Eberlein \& Keller(1995)}]{Eberlein_Ernst_And_Keller_Ulrich_1995}
Eberlein, E. \& Keller, U. (1995).
\newblock Hyperbolic distributions in finance.
\newblock \emph{Bernoulli}, 281--299

\bibitem[{F{\"o}llmer \& Schied(2002)}]{Follmer_Hans_And_Schied_Alexander_2002}
F{\"o}llmer, H. \& Schied, A. (2002).
\newblock Convex measures of risk and trading constraints.
\newblock \emph{Finance and Stochastics}, 6, 429--447

\bibitem[{Frittelli \&
  Gianin(2002)}]{Frittelli_Marco_And_Gianin_Emanuela_Rosazza_2002}
Frittelli, M. \& Gianin, E. R. (2002).
\newblock Putting order in risk measures.
\newblock \emph{Journal of Banking \& Finance}, 26, 1473--1486

\bibitem[{Frittelli \&
  Gianin(2005)}]{Frittelli_Marco_And_Gianin_Emanuela_Rosazza_2005}
Frittelli, M. \& Gianin, E. R. (2005).
\newblock Law invariant convex risk measures.
\newblock Advances in Mathematical Economics, 33--46. Springer

\bibitem[{Hammerstein(2010)}]{Hammerstein_EAv_2010}
Hammerstein, E. (2010).
\newblock Generalized hyperbolic distributions: Theory and applications to
  {CDO} pricing.
\newblock Ph.D. thesis

\bibitem[{Heath(2000)}]{Heath_David_2000}
Heath, D. (2000).
\newblock Back to the future. {P}lenary lecture.
\newblock First World Congress of the Bachelier Finance Society, Paris

\bibitem[{Hellmich \& Kassberger(2011)}]{Hellmich_And_Kassberger_2011}
Hellmich, M. \& Kassberger, S. (2011).
\newblock Efficient and robust portfolio optimization in the multivariate
  generalized hyperbolic framework.
\newblock \emph{Quantitative Finance}, 11, 1503--1516

\bibitem[{Jorion(1996)}]{Jorion_1996}
Jorion, P. (1996).
\newblock Risk2: Measuring the risk in value at risk.
\newblock \emph{Financial Analysts Journal}, 52, 47--56

\bibitem[{Kolm \emph{et~al.}(2014)Kolm, T{\"u}t{\"u}nc{\"u}, \&
  Fabozzi}]{Kolm_And_Tutuncu_And_Fabozzi_2014}
Kolm, P. N., T{\"u}t{\"u}nc{\"u}, R., \& Fabozzi, F. J. (2014).
\newblock 60 years of portfolio optimization: Practical challenges and current
  trends.
\newblock \emph{European Journal of Operational Research}, 234, 356--371

\bibitem[{Konno \emph{et~al.}(1993)Konno, Shirakawa, \&
  Yamazaki}]{Konno_And_Shirakawa_And_Yamazaki_1993}
Konno, H., Shirakawa, H., \& Yamazaki, H. (1993).
\newblock A mean-absolute deviation-skewness portfolio optimization model.
\newblock \emph{Annals of Operations Research}, 45, 205--220

\bibitem[{Konno \& Suzuki(1995)}]{Hiroshi_Ken_1995}
Konno, H. \& Suzuki, K. (1995).
\newblock A mean--variance-skewness portfolio optimization model.
\newblock \emph{Journal of the Operation Research Society of Japan}, 38

\bibitem[{Kozubowski \&
  Podg{\'o}rski(2001)}]{Kozubowski_Tomasz_J_And_Podgorski_Krzysztof_2001}
Kozubowski, T. J. \& Podg{\'o}rski, K. (2001).
\newblock Asymmetric {L}aplace laws and modeling financial data.
\newblock \emph{Mathematical and Computer Modelling}, 34, 1003--1021

\bibitem[{Kozubowski \&
  Rachev(1994)}]{Kozubowski_Tomasz_J_And_Rachev_Svetlozar_T_1994}
Kozubowski, T. J. \& Rachev, S. T. (1994).
\newblock The theory of geometric stable distributions and its use in modeling
  financial data.
\newblock \emph{European Journal of Operational Research}, 74, 310--324

\bibitem[{Kusuoka(2001)}]{Kusuoka_Shigeo_2001}
Kusuoka, S. (2001).
\newblock On law invariant coherent risk measures.
\newblock Advances in Mathematical Economics, 83--95. Springer

\bibitem[{Landsman(2008)}]{Landsman_2008}
Landsman, Z. (2008).
\newblock Minimization of the root of a quadratic functional under a system of
  affine equality constraints with application to portfolio management.
\newblock \emph{Journal of Computational and Applied Mathematics}, 220,
  739--748

\bibitem[{Landsman \& Makov(2016)}]{Markov_2016}
Landsman, Z. \& Makov, U. (2016).
\newblock Minimization of a function of a quadratic functional with application
  to optimal portfolio selection.
\newblock \emph{Journal of Optimization Theory and Applications}, 308--322

\bibitem[{Landsman \& Valdez(2003)}]{Landsman_And_Valdez_2003}
Landsman, Z. M. \& Valdez, E. A. (2003).
\newblock Tail conditional expectations for elliptical distributions.
\newblock \emph{North American Actuarial Journal}, 7, 55--71

\bibitem[{Lo \& MacKinlay(1997)}]{Lo_And_MacKinlay_1997}
Lo, A. W. \& MacKinlay, A. C. (1997).
\newblock Maximizing predictability in the stock and bond markets.
\newblock \emph{Macroeconomic Dynamics}, 1, 102--134

\bibitem[{Madan \& Seneta(1990)}]{Madan_Dilip_B_And_Seneta_Eugene_1990}
Madan, D. B. \& Seneta, E. (1990).
\newblock The variance gamma ({VG}) model for share market returns.
\newblock \emph{Journal of Business}, 511--524

\bibitem[{Malevergne \&
  Sornette(2006)}]{Malevergne_Yannick_And_Sornette_Didier_2006}
Malevergne, Y. \& Sornette, D. (2006).
\newblock Extreme financial risks: From dependence to risk management.
\newblock Springer-Verlag

\bibitem[{Markowitz(1959)}]{Markowitz_Harry_1959}
Markowitz, H. M. (1959).
\newblock Portfolio Selection: Efficient Diversification of Investments,
  volume~16

\bibitem[{McNeil \emph{et~al.}(2015)McNeil, Frey, \&
  Embrechts}]{McNeil_Alexander_J_And_Frey_Rudigger_And_Embrechts_Paul_2015}
McNeil, A. J., Frey, R., \& Embrechts, P. (2015).
\newblock Quantitative Risk Management: Concepts, Techniques and Tools--Revised
  Edition.
\newblock Princeton University Press

\bibitem[{Meng \& Rubin(1993)}]{Meng_Xiao-Li_And_Rubin_Donald_B_1993}
Meng, X. L. \& Rubin, D. B. (1993).
\newblock Maximum likelihood estimation via the {ECM} algorithm: A general
  framework.
\newblock \emph{Biometrika}, 80, 267--278

\bibitem[{Mittnik \& Rachev(1993)}]{Mittnik_Stefan_And_Rachev_Svetlozar_T_1993}
Mittnik, S. \& Rachev, S. T. (1993).
\newblock Modeling asset returns with alternative stable distributions.
\newblock \emph{Econometric Reviews}, 12, 261--330

\bibitem[{Owadally(2011)}]{Owadally_2011}
Owadally, I. (2011).
\newblock An improved closed-form solution for the constrained minimization of
  the root of a quadratic functional.
\newblock \emph{Journal of Computational and Applied Mathematics}, 4428--4435

\bibitem[{Owen \& Rabinovitch(1983)}]{Owen_And_Rabinovitch_1983}
Owen, J. \& Rabinovitch, R. (1983).
\newblock On the class of elliptical distributions and their applications to
  the theory of portfolio choice.
\newblock \emph{The Journal of Finance}, 38, 745--752

\bibitem[{Prause \emph{et~al.}(1999)}]{Prause_Karsten_And_Others_1999}
Prause, K. \emph{et~al.} (1999).
\newblock The generalized hyperbolic model: Estimation, financial derivatives,
  and risk measures.
\newblock Ph.D. thesis

\bibitem[{Pritsker(1997)}]{Pritsker_1997}
Pritsker, M. (1997).
\newblock Evaluating value at risk methodologies: Accuracy versus computational
  time.
\newblock \emph{Journal of Financial Services Research}, 12, 201--242

\bibitem[{Rachev \emph{et~al.}(2005)Rachev, Stoyanov, Biglova, \&
  Fabozzi}]{Rachev_And_Stoyanov_And_Biglova_And_Fabozzi_2005}
Rachev, S. T., Stoyanov, S. V., Biglova, A., \& Fabozzi, F. J. (2005).
\newblock An empirical examination of daily stock return distributions for {US}
  stocks.
\newblock Data Analysis and Decision Support, 269--281. Springer

\bibitem[{Rockafellar \&
  Uryasev(2000)}]{Rockafellar_R_Tyrrell_And_Uryasev_Stanislav_2000}
Rockafellar, R. T. \& Uryasev, S. (2000).
\newblock Optimization of conditional value-at-risk.
\newblock \emph{Journal of risk}, 2, 21--42

\bibitem[{Rockafellar \&
  Uryasev(2002)}]{Rockafellar_R_Tyrrell_And_Uryasev_Stanislav_2002}
Rockafellar, R. T. \& Uryasev, S. (2002).
\newblock Conditional value-at-risk for general loss distributions.
\newblock \emph{Journal of Banking \& Finance}, 26, 1443--1471

\bibitem[{Schoutens(2003)}]{Schoutens_Wim_2003}
Schoutens, W. (2003).
\newblock L{\'e}vy processes in finance: Pricing financial derivatives.
\newblock Wiley Online Library

\bibitem[{Shi \& Kim(2021)}]{Shi_And_Kim_2021}
Shi, X. \& Kim, Y. S. (2021).
\newblock Coherent risk measures and normal mixture distributions with
  applications in portfolio optimization.
\newblock \emph{International Journal of Theoretical and Applied Finance
  (IJTAF)}, 24, 1--18

\bibitem[{Shirvani \emph{et~al.}(2021{\natexlab{a}})Shirvani, Rachev, \&
  Fabozzi}]{ShirvaniM_2021}
Shirvani, A., Rachev, T. S., \& Fabozzi, F. J. (2021{\natexlab{a}}).
\newblock Multiple subordinated modeling of asset returns: Implications for
  option pricing.
\newblock \emph{Rev Quant Finan Acc}, 156, 1329--1342

\bibitem[{Shirvani \emph{et~al.}(2021{\natexlab{b}})Shirvani, Stoyanov, \&
  Fabozzi}]{Shirvani_2021}
Shirvani, A., Stoyanov, S., \& Fabozzi, F. e. a. (2021{\natexlab{b}}).
\newblock Equity premium puzzle or faulty economic modelling?
\newblock \emph{Rev Quant Finan Acc}, 156, 1329--1342

\bibitem[{Uryasev(2000)}]{Uryasev_Stanislav_2000}
Uryasev, S. (2000).
\newblock Conditional value-at-risk: Optimization algorithms and applications.
\newblock Proceedings of the IEEE/IAFE/INFORMS 2000 Conference on Computational
  Intelligence for Financial Engineering (CIFEr) (Cat. No. 00TH8520), 49--57.
  IEEE

\bibitem[{Zhao \emph{et~al.}(2015)Zhao, Lu, Han, Liu, \&
  Hu}]{Zhao_Shangmei_And_Lu_Qing_And_Han_Liyan_And_Liu_Yong_And_Hu_Fei_2015}
Zhao, S., Lu, Q., Han, L., Liu, Y., \& Hu, F. (2015).
\newblock A mean-{CVaR}-skewness portfolio optimization model based on
  asymmetric {L}aplace distribution.
\newblock \emph{Annals of Operations Research}, 226, 727--739

\end{thebibliography}
\end{document}